\documentclass[preprint,3p,11pt]{elsarticle}

\usepackage{amsmath,amsthm,amssymb,amsfonts,xcolor}
\usepackage[caption=false,font=footnotesize]{subfig}

\newtheorem{thm}{Theorem}
\newtheorem{cor}[thm]{Corollary}
\newtheorem{lem}[thm]{Lemma}

\theoremstyle{definition}
\newtheorem{defn}[thm]{Definition}
\theoremstyle{remark}

\newproof{pf}{Proof}

\newcommand{\R}{\mathbb{R}}

\renewcommand{\L}{\mathcal{L}}
\newcommand{\V}{\mathcal{V}}
\newcommand{\X}{\mathcal{X}}
\newcommand{\M}{\mathcal{M}}

\newcommand{\one}{{\bf 1}}

\renewcommand{\ker}{{\rm ker\,}}
\renewcommand{\span}{{\rm span\,}}
\newcommand{\rank}{{\rm rank\,}}
\newcommand{\diag}{{\rm diag\,}}

\journal{}

\begin{document}

\begin{frontmatter}

\title{Graph Fourier Transform Based on $\ell_1$ Norm Variation Minimization}

\author[sysu,gpkl]{Lihua Yang}
\ead{mcsylh@mail.sysu.edu.cn}
\author[sysu]{Anna Qi}
\ead{1350561656@qq.com}
\author[szu]{Chao Huang}
\ead{hchao@szu.edu.cn}
\author[sysu]{Jianfeng Huang\corref{cor}}
\ead{huangjf29@mail.sysu.edu.cn}
\address[sysu]{School of Mathematics, Sun Yat-sen University, Guangzhou 510275, China}
\address[gpkl]{
Guangdong Provincial Key Laboratory of Computational Science, Sun Yat-sen University, Guangzhou 510275, China}
\cortext[cor]{Corresponding author}
\address[szu]{College of Mathematics and Statistics, Shenzhen University, Shenzhen 518060, China}

\begin{abstract}
  The definition of the graph Fourier transform is a fundamental issue in graph signal processing.
  Conventional graph Fourier transform is defined through the eigenvectors of the graph Laplacian matrix, which minimize the $\ell_2$ norm signal variation.
  However, the computation of Laplacian eigenvectors is expensive when the graph is large.
  In this paper, we propose an alternative definition of graph Fourier transform based on the $\ell_1$ norm variation minimization.
  We obtain a necessary condition satisfied by the $\ell_1$ Fourier basis, and provide a fast greedy algorithm to approximate the $\ell_1$ Fourier basis.
  Numerical experiments show the effectiveness of the greedy algorithm.
  Moreover, the Fourier transform under the greedy basis demonstrates a similar rate of decay to that of Laplacian basis for simulated or real signals.
\end{abstract}


\end{frontmatter}


\section{Introduction}
\subsection{Graph Fourier transform}
In many applications such as social, transportation, sensor and neural networks, high-dimensional data is usually defined on the vertices of weighted graphs \cite{Shuman2013}.
To process signals on graphs, traditional theories and methods established on the Euclidean domain need to be extended to the graph setting.
There are many works in this area in recent years, including spectral graph theory \cite{Chung1997},
Fourier transform for directed graphs \cite{Sardellitti2017, Sandryhaila2013}, short-time Fourier transform on graphs \cite{Shuman2016}, wavelets on graphs \cite{Gavish2010, Hammond2011, XChen2014, Dong2017}, graph sampling theory \cite{Chen2015}, uncertainty principle \cite{Agaskar2013}, etc.

The definition of the graph Fourier transform plays a central role in graph signal processing.
By Fourier transform, a graph signal is decomposed into different spectral components and thus can be analyzed from the Fourier domain.
The popular definition of graph Fourier transform is through the eigenvectors of the graph Laplacian matrix.
Although this definition is adopted by many researchers, it has some limitations.
First, the definition only applies to undirected graphs.
Second, the computation of the Laplacian eigenvectors is rather expensive when the graph is large.
Therefore, it is tempting to find an alternative definition of graph Fourier transform without these disadvantages.

One basic requirement for the Fourier basis is that the basis vectors should represent a range of different oscillating frequencies.
For a time-domain signal, the classical Fourier transform decomposes it into different frequency components.
Likewise, in the graph setting, one expects the graph Fourier basis to have a similar property, i.e., the basis vectors represent different oscillating frequencies.
Generally speaking, the magnitude of oscillation of a signal can be measured by its variation.
In fact, the $\ell_2$ norm variation of the Laplacian eigenvectors $u_k$ is characterized by the corresponding eigenvalue $\lambda_k$.
When the eigenvalues $\lambda_k$ are arranged in ascending order, the variation of the eigenvector $u_k$ will be ascending with $k$, thus representing a range of frequencies from low to high.
Moreover, the eigenvector $u_k$ minimizes the $\ell_2$ norm variation in the subspace orthogonal to the span of the previous $k-1$ eigenvectors.

Recently, Sardellitti et al. proposed a definition of directed graph Fourier basis as the set of $N$ orthogonal vectors minimizing the graph directed variation, and proposed two algorithms (SOC and PAMAL) to solve the related optimization problem \cite{Sardellitti2017}.
However, there is a lack of theoretic analysis of the proposed Fourier basis, and the computational complexity of the proposed algorithms are rather high.
Slightly different from Sardellitti's approach, we propose a definition of $\ell_1$ Fourier basis based on iteratively solving a sequence of $\ell_1$ norm variation minimization problems.
We rigorously prove a necessary condition satisfied by the proposed $\ell_1$ Fourier basis.
Further, we provide a fast greedy algorithm to approximately construct the $\ell_1$ Fourier basis.
Numerical experiments show the algorithm is effective, and the Fourier coefficients under the greedy basis and Laplacian basis have nearly the same rate of decay for simulated or real signals.

The rest of the paper is organized as follows.
In Section 2, we discuss the relation between graph Fourier basis and signal variation, and propose the definition of $\ell_1$ Fourier basis based on $\ell_1$ norm variation minimization.
In Section 3, we prove a necessary condition of $\ell_1$ Fourier basis, showing that the $k$th basis vector $u_k$'s components have at most $k$ different values.
In Section 4, we provide a greedy algorithm to construct an approximate $\ell_1$ basis.
In Section 5, we present some numerical results.
Section 6 is a final conclusion.

\subsection{Notations}
In this paper we use the following notations.

For a matrix $M\in\R^{m\times n}$, $\span M$ denotes its column space, i.e., $\{Mx\mid x\in\R^n\}$; and $\ker M$ denotes its kernel, i.e., $\{x\in\R^n\mid Mx = 0\}$.

For a vector $x=[x_1,\dots,x_n]^\top\in\R^n$, $\|x\|$ denotes its Euclidean norm, i.e., $\|x\| = (\sum_{i=1}^n |x_i|^2)^{1/2}$.
For a matrix $M$, $\|M\|$ denotes its operator norm, i.e., $\sup_{x\ne 0}\frac{\|Mx\|}{\|x\|}$. Denote by $B(x,\varepsilon):=\{x'\mid\|x-x'\|<\varepsilon\}$ the open ball centered at $x$ with radius $\varepsilon>0$.

The cardinality of a set $A$ is denoted by $|A|$.
Let $N$ be a positive integer, and $\V = \{1,\dots,N\}$.
For any $A\subset \V$, we use $\one_A\in\R^N$ to denote the indication vector of $A$, i.e., $\one_A(i) = 1$ if $i\in A$ and $\one_A(i)=0$ otherwise.
$\one_\V$ is also written as $\one$.

For $W=[w_{ij}]\in\R^{N\times N}$ and subsets $A, B\subset \{1,\cdots, N\}$, $W(A,B)$ is defined as $\sum\limits_{i\in A}\sum\limits_{j\in B}w_{ij}$.

\section{Graph Fourier basis and signal variation}

In this section, we shall derive the relationship between the graph Fourier basis and signal variation.
Let us begin with the basic terminology of graph signal processing.
Let $G=(\V,W)$ be a connected, undirected, and weighted graph, where the vertices set $\V=\{1,2,\dots, N\}$ and the weight matrix $W=[w_{ij}] \in R^{N\times N}$ satisfying $w_{ij}=w_{ji}\ge 0$ and $w_{ii}=0$.
The degree of a vertex is defined as $d_i = \sum_{j=1}^N w_{ij}$, and the degree matrix $D=\diag(d_1,\dots,d_N)$.
The combinatorial Laplacian matrix is defined as $\L = D-W$.
Since $\L$ is symmetric and positive semi-definite, it has eigenvalues $0=\lambda_1\le \cdots\le \lambda_N$ and the corresponding set of orthonormal eigenvectors $\{u_1,\dots,u_N\}$.
We call $U=[u_1,\dots,u_N]\in\R^{N\times N}$ the Laplacian basis of $G$.
A graph signal $x$ is a real-valued function defined on $\V$, and can be regarded as a vector in $\R^N$.
The Fourier transform of $x$ under the Laplacian basis is defined as $U^\top x$.

Note that the $\ell_2$ norm variation of the Laplacian eigenvector $u_k$ is increasing with $k$.
To see this, let $x=[x_1,\dots,x_N]^\top\in\R^N$, then it can be proved that
\begin{equation}\label{eq_1}
    x^\top \L x = \sum_{1\le i < j \le N} w_{ij}|x_i - x_j|^2.
\end{equation}
That means the quadratic form $x^\top\L x$ exactly measures the $\ell_2$ norm variation of $x$.
Since $u_k^\top \L u_k=\lambda_k$, we have
\[
    u_1^\top \L u_1 \le \cdots \le u_N^\top \L u_N,
\]
i.e., the $\ell_2$ norm variation of $u_k$ is increasing with $k$.
In other words, the Laplacian basis vectors $\{u_k\mid 1\le k\le N\}$ represent a range of frequencies from low to high.

Furthermore, the eigenvector $u_k$ minimizes the $\ell_2$ norm variation in the subspace orthogonal to the span of the previous $k-1$ eigenvectors, i.e.,
\begin{equation}\label{eq_four2}
    \begin{array}{lcl}
     u_k = &\mathop{\arg\,\min}\limits_{x\in\R^N} & x^\top \L x\\
    &\text{s. t.}& [u_1,\dots, u_{k-1}]^\top x = 0, \ \|x\|=1.
    \end{array}
\end{equation}
In fact, let $x\in\R^N$ satisfy $[u_1,\dots, u_{k-1}]^\top x=0$ and $\|x\|=1$. Let the Fourier transform of $x$ be $\hat x=U^\top x=[\hat x_1,\dots,\hat x_N]^\top$.
Then $x$ can be expressed as $\sum_{j=k}^N \hat x_j u_j$, hence
\[x^\top \L x = \hat x^\top U^\top \L U \hat x = \sum_{j=k}^N \lambda_j|\hat x_j|^2 \ge \lambda_k \sum_{j=k}^N|\hat x_j|^2 = \lambda_k = u_k^\top \L u_k.\]
Therefore the eigenvector $u_k$ solves the $\ell_2$ norm variation minimization problem (\ref{eq_four2}) for $k=2,\dots, N$.

It is natural to consider the more general $\ell_p$ norm variation.
In this paper, we restrict ourselves to $\ell_1$ norm variation defined as follows
\begin{equation}
    S(x):= \sum_{1\le i < j\le N}w_{ij}|x_i-x_j|.
\end{equation}

Similar to Laplacian basis minimizing $\ell_2$ norm variation, we define the $\ell_1$ Fourier basis as the solution of $\ell_1$ norm variation minimization problem.
\begin{defn}\label{def_l1}
    Let $u_1:=\frac{\one}{\sqrt{N}}$. If a sequence of vectors $\{u_k\mid 2\le k\le N\}$ solves the $\ell_1$ norm variation minimization problem as follows,
    \begin{equation}\label{eq_four1}
        \begin{array}{lcl}
        u_k = & \mathop{\arg\,\min}\limits_{x\in\R^N} & S(x)\\
        & \text{s. t.}& [u_1,\dots, u_{k-1}]^\top x = 0, \ \|x\|=1.
        \end{array}
    \end{equation}
    for $k=2,\dots,N$,
    then we say the orthogonal matrix $U=[u_1,\dots,u_N]\in\R^{N\times N}$ constitutes an $\ell_1$ Fourier basis, or simply an $\ell_1$ basis, of the graph $G$.
\end{defn}

Remarks:
The above definition of $\ell_1$ Fourier basis can be extended to directed graphs. All one needs is to replace $S(x)$ in the minimization problem by a directed version
\begin{equation}
    \widetilde S(x):= \sum_{1\le i,j\le N}w_{ij}(x_i-x_j)_+,
\end{equation}
where $(x_i-x_j)_+ = \max (x_i-x_j, 0)$ (more details can be found in \cite{Sardellitti2017}).
Then one can similarly defined the directed $\ell_1$ Fourier basis as the solution of the corresponding problem.
Without loss of generality, we only consider undirected graphs in this paper.
Most results can be generated to the directed case without essential difficulties.

\section{Necessary condition of $\ell_1$ Fourier basis}
In the previous section, the $\ell_1$ Fourier basis vectors are defined as the solutions of a sequence of minimization problem (\ref{eq_four1}).
We rewrite problem (\ref{eq_four1}) in a concise form:
\begin{equation}
    P_U:=\quad
    \begin{array}{cl}
    \min\limits_{x\in\R^N} & S(x)\\
    \text{s. t.}& U^\top x = 0, \quad \|x\|=1
    \end{array}
\end{equation}
where $U\in\R^{N\times(k-1)}$ is a matrix with its first column being $\frac{\one}{\sqrt{N}}$, $\rank(U)=k-1$, and $2\le k\le N$.
With this notation, problem (\ref{eq_four1}) can be referred to as $P_{[u_1,\dots,u_{k-1}]}$.
Now our goal is to solve problem $P_U$.

First, let us recall some basic definitions of optimization theory.
Denote the feasible region of problem $P_U$ by $\X_U$, i.e.,
\begin{equation}
    \X_U:=\{x\in\R^N\mid U^\top x = 0,\ \|x\|=1\}.
\end{equation}
A point $x\in\X_U$ is called a local minimum of problem $P_U$ if there exists $\varepsilon>0$ such that $S(x')\ge S(x)$ for any $x'\in\X_U\cap B(x,\varepsilon)$.
If $S(x')\ge S(x)$ holds for any $x'\in\X_U$, then $x$ is called a global minimum of problem $P_U$.
Obviously a global minimum is necessarily a local minimum.
We denote the set of all local minima of problem $P_U$ by $\X^{**}_U$.

Due to the sphere constraint $\|x\|=1$, problem $P_U$ is not a convex optimization problem.
As far as we know, there are no general results about the global minimum of such problems, and in most cases it is only possible to approach the local minimum by iterative algorithms \cite{Bresson2012, Lai2014}.
As the main result of this section, we shall prove a necessary condition satisfied by the local minimum (Theorem \ref{thm_1}).
The key ingredient of the proof is based on the concept of piecewise representation, which is introduced as follows.

\begin{defn}
Suppose $x=[x_1,\dots, x_N]^\top\in\R^N$. Let $X:=\{x_i\mid 1\le i\le N\}$ and $m:=|X|$.
Then $X$ can be rewritten as $\{x_{(j)}\mid 1\le j\le m\}$, where $x_{(1)} < x_{(2)} < \cdots < x_{(m)}$.
Let $A_j:=\{i\mid 1\le i\le N,\ x_i=x_{(j)}\}$, $M:=[\one_{A_1},\dots,\one_{A_m}]\in\mathbb{R}^{N\times m}$ and $a:=[x_{(1)},\dots, x_{(m)}]^\top\in\R^m$.
Then $x=Ma$, which is called the \emph{piecewise representation} of $x$. We also call $M$ the \emph{partition matrix} of $x$, denoted by $\phi(x)=M$.
\end{defn}
\par
It is easy to see that any vector in $\R^N$ has unique piecewise representation.
Under the piecewise representation $x=Ma$, the $\ell_1$ norm variation $S(x)$ can be simplified to a linear form in a local neighborhood of $a$.

\begin{lem}\label{lem_f}
Suppose $x\in\R^N$, $\phi(x)=M=[\one_{A_1},\dots,\one_{A_m}]$, $x=Ma$ and $m\ge 2$.
Then there exists $\varepsilon>0$ such that
  \begin{equation}
    S(Ma') = f^\top a',\quad \forall a'\in B(a,\varepsilon),
  \end{equation}
where $f=[f_1,\cdots,f_m]^\top\in\R^m$ is defined by
\begin{equation}
    f_i:=\sum^{i-1}_{j=1}W(A_i,A_j)-\sum^m_{j=i+1}W(A_i,A_j), \quad i=1,\dots,m.
\end{equation}
\end{lem}

\begin{proof}
Suppose $a=[a_1,\dots,a_m]^\top$, then $a_1<\cdots<a_m$.
Let $a'=[a'_1,\dots,a'_m]^\top$ and $x'=Ma'$ .
When $\|a-a'\|$ is sufficiently small, we have $a'_1<\cdots<a'_m$, i.e., there exists $\varepsilon>0$ such that for all $a'\in B(a,\varepsilon)$, $x'=Ma'$ is a piecewise representation. Therefore
\begin{eqnarray*}
    S(x') &=& \frac{1}{2}\sum_{i=1}^N\sum_{j=1}^N w_{i,j}|x_i'-x_j'|
    \\&=& \frac{1}{2}\sum_{i=1}^m\sum_{p\in A_i}\sum_{j=1}^m\sum_{q\in A_j} w_{p,q}|x_p'-x_q'|
    \\&=& \frac{1}{2}\sum_{i=1}^m\sum_{j=1}^m|a_i'-a_j'|\sum_{p\in A_i}\sum_{q\in A_j}w_{p,q}
    \\&=& \sum_{1\le i<j\le m}(a_j'-a_i')W(A_i,A_j)
    \\&=& f^\top a'
\end{eqnarray*}
\end{proof}

\begin{thm}\label{thm_1}
If $x\in\X^{**}_U$ and $\phi(x)=M$, then
    \begin{equation}\label{eq_nec}
        \dim\ker (U^\top M) = 1.
    \end{equation}
\end{thm}
\begin{proof}
The main idea is to transform problem $P_U$ to a easier one by using Lemma \ref{lem_f}.
Suppose $M=[\one_{A_1},\dots,\one_{A_m}]$ and $x=Ma$.
By assumption of problem $P_U$, we have $\langle x,\one\rangle=0$ and $\|x\|=1$, therefore $x$ is a non-constant signal, i.e., $m\ge 2$.
Since $x$ is a local minimum of $P_U$, there exists $\varepsilon_1 > 0$ such that
\begin{equation}
    \begin{array}{lcl}
    x\ = & \mathop{\rm arg\,min}\limits_{x'\in\R^N} & S(x')\\
    &\text{s. t.}& U^\top x' = 0, \ \|x'\| = 1,\ x'\in B(x,\varepsilon_1).
    \end{array}
\end{equation}

By Lemma \ref{lem_f}, there exists $\varepsilon_2>0$ and $f\in\R^m$ such that $S(x')=f^\top a'$ for all $a'\in B(a,\varepsilon_2)$ and $x'=Ma'$.
Let $\varepsilon:=\min\{\varepsilon_1/\|M\|,\varepsilon_2\}$, then $a'\in B(a,\varepsilon)$ implies $x'\in B(x,\varepsilon_1)$ and $S(x')=f^\top a'$.
Let $\Lambda:=\diag(|A_1|,\dots,|A_m|) = M^\top M$, then $a'^\top\Lambda a'=1$ implies $\|x'\|=1$, and $U^\top Ma'=0$ implies $U^\top x' = 0$.
Therefore
\begin{equation}\label{pa}
    \begin{array}{lcl}
    a\ = & \mathop{\rm arg\,min}\limits_{a'\in\R^m} & f^\top a'\\
    &\text{s. t.}& U^\top M a' = 0, \ a'^\top \Lambda a' = 1,\ a'\in B(a,\varepsilon).
    \end{array}
\end{equation}

Suppose $\dim\ker(U^\top M) = l$, and $V$ is an orthonormal basis of $\ker (U^\top M)$.
Define $c:=V^\top a$, $g^\top := f^\top V$, $Q := V^\top \Lambda V$.
Then we have
\begin{equation}\label{eq_Pc}
    \begin{array}{lcl}
    c\ = & \mathop{\rm arg\,min}\limits_{c'\in\R^l} & g^\top c'\\
    &\text{s. t.}& c'^\top Q c' = 1,\ c'\in B(c,\varepsilon).
    \end{array}
\end{equation}
We next prove problem (\ref{eq_Pc}) has minimum only if $l=1$.
It is proved by contradiction.

Suppose $l\ge 2$.
By the method of Lagrange multipliers, the minimum $c$ of problem (\ref{eq_Pc}) satisfies the equation \[
    \nabla[g^\top c + \mu(c^\top Q c - 1)] = g+2\mu Q c = 0,
\]
where $\mu$ is a Lagrange multiplier. Thus $g=-2\mu Qc$.

Since $l\ge 2$, there exists a nonzero vector $r'\in\R^l$ such that $r'^\top c = 0$.
Let $r := Q^{-1} r'$, then $g^\top r = -2\mu c^\top Q r = -2\mu c^\top r' = 0$.
Let $c'':=c+tr$, $t\in\R$, $t\ne 0$. Then
\[
    c''^\top Q c'' = c^\top Q c+ 2 t c^\top Q r + t^2 r^\top Qr = 1 + t^2 r^\top Q r > 1,
\]
since $Q$ is symmetric and positive definite.

Let $c':=c''/\sqrt{c''^\top Q c''}$, then $c'^\top Q c' = 1$. Choose $|t|$ small enough to guarantee $c'\in B(c,\varepsilon)$. Since $g^\top c = f^\top a = S(x) > 0$, we have
\[
    g^\top c' = \frac{g^\top c + t g^\top r}{\sqrt{c''^\top Q c''}} = \frac{g^\top c}{\sqrt{c''^\top Q c''}} < g^\top c,
\]
which contradicts to $c$ being the minimum of problem (\ref{eq_Pc}). The proof is complete.
\end{proof}

We remark that the condition (\ref{eq_nec}) is not a sufficient condition.
From condition (\ref{eq_nec}), we deduce an estimate of the number of values of the components of a local minimum $x$.
\begin{cor}\label{cor}
    If $x\in\X_U^{**}$, then the components of $x$ have at most $k$ different values.
\end{cor}
\begin{proof}
Let $\phi(x)=M\in\R^{N\times m}$. By Theorem \ref{thm_1}, $\dim\ker(U^\top M) = 1$.
Since
  \begin{eqnarray*}
    k-1&=&\rank(U)
    \\&\ge& \rank (M^\top U)
    \\&=& \dim\span (M^\top U)
    \\&=& m-\dim\ker (U^\top M)
    \\&=&m-1,
  \end{eqnarray*}
we have $m\le k$.
By definition of piecewise representation, $m$ is the number of different values of $x$'s components. The proof is complete.
\end{proof}

Corollary \ref{cor} asserts that the $k$th $\ell_1$ basis vector $u_k$, as the global (hence local) minimum of problem $P_{[u_1,\dots,u_{k-1}]}$, is at most a $k$-valued signal.
In particular, $u_1$ is a constant signal and $u_2$ is exactly a two-valued signal.
Intuitively speaking, the larger $k$ is, the more values $u_k$ can take, the more oscillation $u_k$ might present.
Thus the $\ell_1$ basis vectors $\{u_k\}$ represent different oscillation frequencies from low to high as expected.

Another implication of condition (\ref{eq_nec}) is the finiteness of the set of local minima.
Denote by $\M_U^*$ the set of all partition matrices of $x\in\X_U$ satisfying condition (\ref{eq_nec}), i.e.,
\begin{equation}
    \M_U^*:=\{M\mid x\in\X_U,\ M=\phi(x),\ \dim\ker (U^\top M)=1\}.
\end{equation}
For any vector $x\in\R^N$, its partition matrix $M$ has at most $N$ columns, and each entry of $M$ is either $0$ or $1$. Therefore the set of all partition matrices of vectors in $\R^N$ is a finite set, so $\M^*_U$ as a subset is also finite.

By Theorem \ref{thm_1}, if $x$ is a local minimum of problem $P_U$, then its partition matrix belongs to $\M^*_U$.
Conversely, given a partition matrix $M\in\M^*_U$, we show that there are only two $x\in\X_U$ with partition matrix being equal to $M$.

\begin{thm}\label{thm_2}
If $M\in\M_U^*$ and $x, x'\in\X_U\cap\span M$, then $x=\pm x'$.
\end{thm}
\begin{proof}
    Since $x,x'\in\X_U\cap\span M$, there exist $a,a'$ such that $x=M a$ and $x'=M a'$.
    Then $U^\top Ma=U^\top x = 0$ and $U^\top Ma'=U^\top x' = 0$, i.e., $a, a'\in\ker(U^\top M)$.
    Since $\dim\ker(U^\top M)=1$ and $a,a'\ne 0$, there exists $t\in\R$ such that $a=t a'$, hence $x=t x'$.
    From $\|x\|=\|x'\|=1$, we have $t=\pm 1$.
    The proof is complete.
\end{proof}

Define
\begin{equation}
    \psi_U (M):= \{x\mid x\in\X_U\cap\span M\},\quad \forall M\in\M_U^*.
\end{equation}
By Theorem \ref{thm_2}, $\psi_U(M)$ has two elements in total, which differs by a sign.
Let
\begin{equation}
    \X^*_U := \bigcup_{M\in\M_U^*}\psi_U(M).
\end{equation}
Then $|\X^*_U|\le \sum_{M\in\M^*_U}|\psi_U(M)|=2|\M_U^*|<\infty$, i.e., $\X^*_U$ is a finite set.

The local minima set $\X^{**}_U$ is a subset of $X^*_U$.
In fact, If $x\in\X^{**}_U$ and $\phi(x)=M$, then $M\in\M_U^*$ and $x\in\X_U\cap\span M$, hence $x\in\psi_U(M)\subset \X^*_U$.
It follows that $\X^{**}_U$ is also a finite set, i.e., each local minima is isolated and the total number of local minima is finite.
Figure \ref{fig_T} shows the relations between these sets and definitions.
Here $\X^*_U$ resembles the concept of critical points, which contains but not equals the set of local minima.
\begin{figure}[!h]
    \centering
    \includegraphics[width=.55\textwidth]{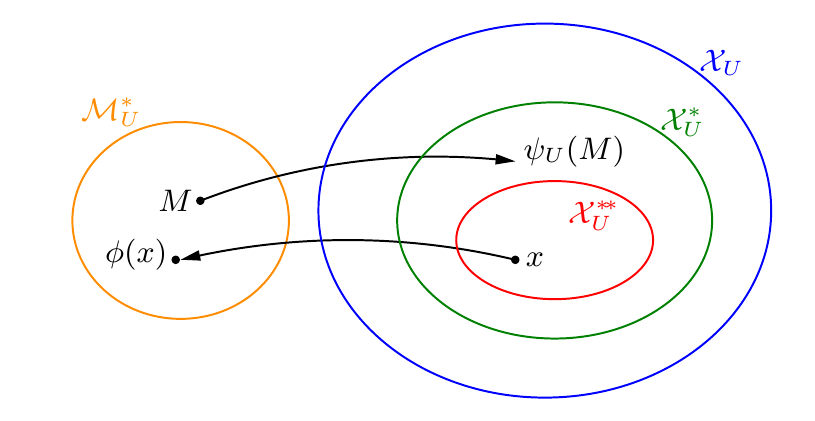}\\
    \caption{Relation between $\M^*_U$, $\X^{**}_U$, $\X^*_U$ and $\X_U$.}\label{fig_T}
\end{figure}

Since $\X^*_U$ is finite, to find the global minimum of problem $P_U$, one way is to compute $S(x)$ for all $x$ in $\X^*_{U}$ and pick out the largest one.
Table \ref{tab_1} shows a special case for $N=4$, $U=\frac{\one}{\sqrt{N}}$.
\begin{table}[!h]
    \begin{center}
    \begin{tabular}{p{3cm} p{3.5cm} p{4.5cm}}
        \hline
        $M\in\M_U^*$ & $\pm x\in\psi_U (M)$ & $S(x)$\\
        \hline
        $[\one_{\{1\}},\one_{\{2,3,4\}}]$ & $\pm\frac{1}{2\sqrt{3}}[-3,1,1,1]^\top$ & $\frac{2}{\sqrt{3}}(w_{12}+w_{13}+w_{14})$ \\   $[\one_{\{2\}},\one_{\{1,3,4\}}]$ & $\pm\frac{1}{2\sqrt{3}}[1,-3,1,1]^\top$ & $\frac{2}{\sqrt{3}}(w_{12}+w_{23}+w_{24})$ \\
        $[\one_{\{3\}},\one_{\{1,2,4\}}]$ & $\pm\frac{1}{2\sqrt{3}}[1,1,-3,1]^\top$ & $\frac{2}{\sqrt{3}}(w_{13}+w_{23}+w_{34})$ \\
        $[\one_{\{4\}},\one_{\{1,2,3\}}]$ & $\pm\frac{1}{2\sqrt{3}}[1,1,1,-3]^\top$ & $\frac{2}{\sqrt{3}}(w_{14}+w_{24}+w_{34})$ \\
        $[\one_{\{1,2\}},\one_{\{3,4\}}]$ & $\pm\frac{1}{2}[-1,-1,1,1]^\top$ & $w_{13}+w_{14}+w_{23}+w_{24}$ \\
        $[\one_{\{1,3\}},\one_{\{2,4\}}]$ & $\pm\frac{1}{2}[-1,1,-1,1]^\top$ & $w_{12}+w_{14}+w_{23}+w_{34}$ \\
        $[\one_{\{1,4\}},\one_{\{2,3\}}]$ & $\pm\frac{1}{2}[-1,1,1,-1]^\top$ & $w_{12}+w_{13}+w_{23}+w_{24}$ \\
        \hline
     \end{tabular}
     \caption{Enumeration of $x$ in $\X^*_U$ for $N=4$, $U=\frac{\one}{\sqrt{N}}$.} \label{tab_1}
     \end{center}
 \end{table}
Through this method of enumeration, the continuous problem $P_U$ is equivalent to a discrete problem in which the variable $x$ belongs to a finite set $\X^*_{U}$.
However, as far as we know, the discrete problem has no effective algorithm, since the size of $\X^*_U$ grows exponentially with $N$, and the method of enumeration is impractical for large $N$.
In the next section, we will give a fast greedy algorithm to approximately construct the $\ell_1$ Fourier basis when $N$ is large.

\section{Greedy algorithm for $\ell_1$ Fourier basis}

In this section, we provide a fast greedy algorithm to approximately construct the $\ell_1$ Fourier basis.
Through piecewise representation, the partition matrix of the $k$th $\ell_1$ basis vector $u_k$ naturally induces a partition of the vertices set $\V$.
The increasing of variation of $u_k$ implies that the corresponding partition evolves from coarser to finer scales.
On the contrary, given a sequence of partitions varying across different scales, one might be able to construct an orthonormal basis close to $\ell_1$ basis.
Motivated by this idea, we propose a greedy algorithm, based on a partition sequence $\tau_k$ created by iteratively grouping the vertices.
In each step, we pick out the two groups of vertices with the largest mutual weights between them, and combine them in a new group.
Repeating the process, we get a sequence of partitions $\tau_k$ varying from finer to coarser scales.
Then based on $\tau_k$, we define a sequence of subspaces $V_k$ of $\R^N$.
By using the similar ideas of multi-resolution analysis, we obtain an orthonormal basis.

\subsection{Greedy partition sequence}
We define a sequence of partitions $\tau_k$ on the vertices set $\V=\{1,\dots,N\}$ as follows.
\begin{defn}\label{def_tau}
    Let
    \begin{equation}
        \tau_N:= \{\{1\}, \{2\},\dots,\{N\}\}.
    \end{equation}
    For $k=N, N-1,\dots, 2$, define
    \begin{equation}
        A_k, B_k:=\mathop{\arg\,\max}\limits_{A,B\in\tau_{k}} W(A, B),
    \end{equation}
    \begin{equation}
        \tau_{k-1}:=\{A_k\cup B_k\}\cup\{C\in\tau_{k}\mid C\ne A_k,\ C\ne B_k\}.
    \end{equation}
\end{defn}

Definition \ref{def_tau} actually represents a vertices grouping process.
At the beginning, the finest partition $\tau_N$ has $N$ groups, each group having one vertex.
To get the next partition $\tau_{N-1}$, we identify $A_N,B_N$ as the two groups having the largest mutual weight.
Then we combine $A_N$ and $B_N$ to get a new group $A_N\cup B_N$, and together with the other groups in $\tau_N$ we form a new partition $\tau_{N-1}$.
This operation repeats for $N-1$ times.
At the end, we get the coarsest partition $\tau_1 = \{1,2,\dots, N\}$, with all the vertices belonging to a single group.
See Figure \ref{fig_d} for an illustration.
\begin{figure}[!h]
    \centering
    \includegraphics[width=.65\textwidth]{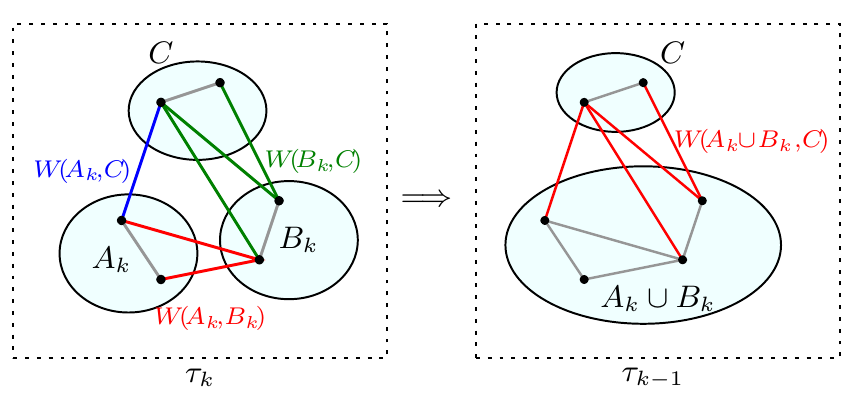}\\
    \caption{In step $k$, we combine $A_k$ and $B_k$ of $\tau_k$ to get $\tau_{k-1}$.}\label{fig_d}
\end{figure}

\subsection{Greedy basis}
The greedy partition sequence $\tau_k$ defined above yields a sequence of subspaces
\begin{equation}
    V_k:=\span \{\one_A\mid A\in\tau_k\},\quad k=1,\dots,N,
\end{equation}
which satisfy the relations
\begin{equation}
    \span \one = V_1 \subset V_2 \subset \cdots\subset V_N = \R^N.
\end{equation}

Denote the orthogonal complement of $V_{k-1}$ in $V_k$ by $V_k\ominus V_{k-1}$.
By definition \ref{def_tau}, the partition $\tau_{k-1}$ is obtained by combining two groups $A_k$ and$B_k$ in $\tau_{k}$.
Suppose $\tau_{k}=\{A_k,B_k, C_1,\dots, C_{k-2}\}$ and $\tau_{k-1} = \{A_k\cup B_k, C_1,\dots, C_{k-2}\}$.
Let $x\in V_k\ominus V_{k-1}$.
Then $x$ can be written in the form $a\one_{A_k}+b\one_{B_k}+\sum c_i \one_{C_i}$.
From $\langle x, \one_{C_i}\rangle =c_i|C_i|=0$, we get $c_i=0$, $\forall i=1,\dots,k-2$.
Since
\[
    \langle x,\one_{A_k\cup B_k}\rangle = a|A_k|+b|B_k| = 0,
\]
there exists $t\in\R$ such that $a=t|B_k|$, $b=-t|A_k|$.
By requiring $\|x\|=1$,  we get $t = \frac{\pm 1}{\sqrt{|A_k||B_k|(|A_k|+|B_k|)}}$.
We summarize these results in the following theorem.

\begin{thm}\label{thm_greedy}
    Suppose $A_k, B_k$ are defined as in Definition \ref{def_tau}. Let $\widetilde u_1 :=\tfrac{\one}{\sqrt{N}}$,
    \begin{equation}
        \widetilde u_k := a_k\one_{A_k}+b_k\one_{B_k}, \quad k=2,\dots, N,
    \end{equation}
    where
    \begin{equation}
        a_k := -t_k|B_k|,\quad b_k := t_k|A_k|,\quad t_k := \frac{1}{\sqrt{|A_k||B_k|(|A_k|+|B_k|)}}.
    \end{equation}
    Then $\widetilde U=[\widetilde u_1,\dots,\widetilde u_N]$ is an orthogonal matrix.
    We call $\widetilde U$ the greedy basis of the graph $G$.
\end{thm}

Table \ref{tab_2} shows a simple example of the greedy basis $\widetilde U$ given a partition sequence $\tau_k$, where the number of vertices $N = 5$.
Figure \ref{fig_btree} plots the binary tree formed by $A_k$ and $B_k$.
\begin{table}[!h]
    \centering
    \begin{tabular}{p{.5cm}p{4.5cm}p{1.5cm}p{1.5cm}p{4.cm}}
    \hline
    $k$&$\tau_k$&$A_k$&$B_k$&$\widetilde u_k$
    \\\hline
    $5$&$\{\{1\},\{2\},\{3\},\{4\},\{5\}\}$&$\{1\}$&$\{3\}$&$\frac{1}{\sqrt{2}}[-1,0,1,0,0]^\top$
    \\$4$&$\{\{1,3\},\{2\},\{4\},\{5\}\}$&$\{2\}$&$\{5\}$&$\frac{1}{\sqrt{2}}[0,-1,0,0,1]^\top$
    \\$3$&$\{\{1,3\},\{2,5\},\{4\}\}$&$\{1,3\}$&$\{4\}$&$\frac{1}{\sqrt{6}}[-1,0,-1,2,0]^\top$
    \\$2$&$\{\{1,3,4\},\{2,5\}\}$&$\{1,3,4\}$&$\{2,5\}$&$\frac{1}{\sqrt{30}}[-2,3,-2,-2,3]^\top$
    \\$1$&$\{\{1,2,3,4,5\}\}$&&&$\frac{1}{\sqrt{5}}[1,1,1,1,1]^\top$
    \\\hline
    \end{tabular}
    \caption{An example of greedy Fourier basis.}\label{tab_2}
\end{table}
\begin{figure}[!h]
  \centering
  \includegraphics[width=.3\textwidth]{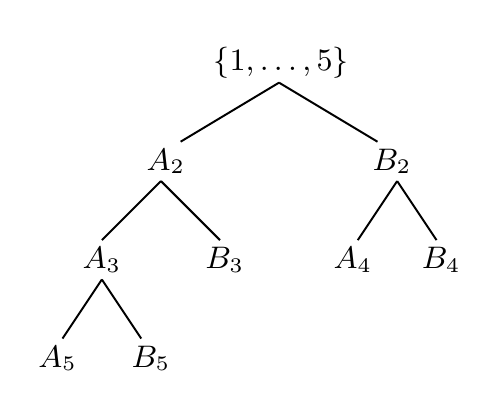}\\
  \caption{Binary tree of $A_k$ and $B_k$ in the above example}\label{fig_btree}
\end{figure}

An interesting question is whether the greedy basis vector $\widetilde u_k$ minimizes the $\ell_1$ norm variation.
We will show that the partition matrix induced by the greedy partition $\tau_k$ satisfies the necessary condition (\ref{eq_nec}).
\begin{thm}\label{thm_6}
    Let
    \begin{equation}
    \widetilde U_{k-1} := [\widetilde u_1, \dots,\widetilde u_{k-1}], \quad k=2, \dots, N,
    \end{equation}
    where $\widetilde u_k$ is defined in Theorem \ref{thm_greedy}.
    Suppose $\tau_k = \{A_k, B_k, C_1, \dots, C_{k-2}\}$, and $M = [\one_{A_k}, \one_{B_k}, \one_{C_1},\\ \dots, \one_{C_{k-2}}]$.
    Then
    $\dim\ker(\widetilde U_{k-1}^\top M) = 1$.
\end{thm}

\begin{proof}
    Suppose $y=[a, b, c_1, \dots, c_{k-2}]^\top\in\ker(\widetilde U_{k-1}^\top M)$ and $x = M y$.
    Then $\widetilde U_{k-1}^\top x =\widetilde U_{k-1}^\top M y = 0$, i.e., $x \bot \span \widetilde U_{k-1}$.
    Since $\span \widetilde U_{k-1} = \span V_{k-1}$, we have $x \bot V_{k-1}$.
    Because $x=My\in \span\{\one_A\mid A\in\tau_k\}=V_k$,  that means $x\in V_k\ominus V_{k-1}$.
    Since $\dim(V_k\ominus V_{k-1})=1$ and $\widetilde u_k\in V_k\ominus V_{k-1}$, there exists $t\in\R$ such that $x=t \widetilde u_k$, i.e., $a\one_{A_k}+b\one_{B_k}+\sum c_i\one_{C_i} = t a_k\one_{A_k}+t b_k\one_{B_k}$.
    Hence $a=ta_k$, $b=tb_k$, $c_i=0$, i.e., $y = t[a_k, b_k, 0, \dots, 0]^\top$,
    therefore $\ker (\widetilde U_{k-1}^\top M) = \span\{[a_k, b_k, 0, \dots, 0]^\top\}$ and $\dim\ker (\widetilde U_{k-1}^\top M) = 1$.
\end{proof}

In Theorem \ref{thm_6}, $M$ and $\widetilde U_{k-1}$ satisfy the condition (\ref{eq_nec}), i.e., $M\in\M_{\widetilde U_{k-1}}^*$.
Since $\widetilde u_k \in\X_{\widetilde U_{k-1}} \cap \span M$, we have $\widetilde u_k\in\X^*_{\widetilde U_{k-1}}$, i.e. $\widetilde u_k$ can be seen as a `critical point' of problem $P_{\widetilde U_{k-1}}$, but not necessarily a local minimum.
Despite of this, the greedy basis $\widetilde U$ provides a rather good approximation to the $\ell_1$ basis, as demonstrated in the numerical experiments later.

\subsection{Fourier transform under the greedy basis}
Let us consider the computation of the Fourier coefficients of a signal $x$ under the greedy basis $\widetilde U$:
\begin{equation}
    \widetilde{\hat x}(k) := \langle x,\widetilde u_k\rangle = \langle x, a_k\one_{A_k} + b_k\one_{B_k}\rangle=a_k \alpha_k + b_k \beta_k
\end{equation}
where
\begin{equation}
    \alpha_k:=\langle x,\one_{A_k}\rangle,\quad \beta_k:=\langle x,\one_{B_k}\rangle
\end{equation}
From Definition \ref{def_tau}, the set of $A_k$'s and $B_k$'s form a binary tree.
Suppose $A_j$ is the parent node of $A_k$ and $B_k$, i.e., $A_j=A_k\cup B_k$, then we have
\begin{equation}
    \alpha_j = \langle x,\one_{A_j}\rangle = \langle x,\one_{A_k}\rangle + \langle x,\one_{B_k}\rangle = \alpha_k + \beta_k.
\end{equation}
Thus the $\alpha_j$'s and $\beta_j$'s also form a binary tree, and can be computed from bottom to up based on the tree structure.
Indeed, the computation of $\widetilde{\hat x}$ needs $O(N)$ multiplications, while the Laplacian basis transform $\hat x$ needs $O(N^2)$ multiplications, since each inner product $\hat x(k)=\langle x,u_k\rangle$ takes $O(N)$ multiplications.
So greedy basis transform is much faster than the Laplacian basis transform.

\section{Numerical Experiments}

\subsection{Error between the greedy basis and $\ell_1$ basis}
In our first experiment, we aim to examine the difference between the greedy basis $\widetilde U$ and the $\ell_1$ basis $U$.
When the vertices number $N$ is small, one can enumerate the finite set $\X^*_U$ to find the global minimum of the $\ell_1$ norm variation.
When $N$ is large, to our knowledge, there is no effective algorithm to obtain the global minimum.
Therefore we restrict $N\le 8$ here so that the accurate $\ell_1$ basis can be obtained by enumeration.

Since $\widetilde u_1$ and $u_1$ are equal, we begin from $u_2$ and $\widetilde u_2$.
Denote the relative error of their variations by
\[
    r(\widetilde u_2,u_2):=\frac{S(\widetilde u_2) - S(u_2)}{S(u_2)}.
\]
In Figure \ref{fig_ru2}(a) the red line plots the average of $r(\widetilde u_2,u_2)$ for $100$ random graphs.
Each of these graphs is generated by $N$ random points $p_i\in\R^2$, and the weights are defined by $w_{ij} := \exp(\|p_i-p_j\|^2/\sigma^2)$ for some parameter $\sigma$.
For the sake of completeness, we also plot the relative error $r(u'_2, u_2)$ in the blue line, where $u_2'$ is the second Laplacian basis vector.
It can be seen that the error $r(\widetilde u_2,u_2)$ is close to zero.
\begin{figure}[!h]
    \centering
    \subfloat[]{\includegraphics[width=.5\textwidth]{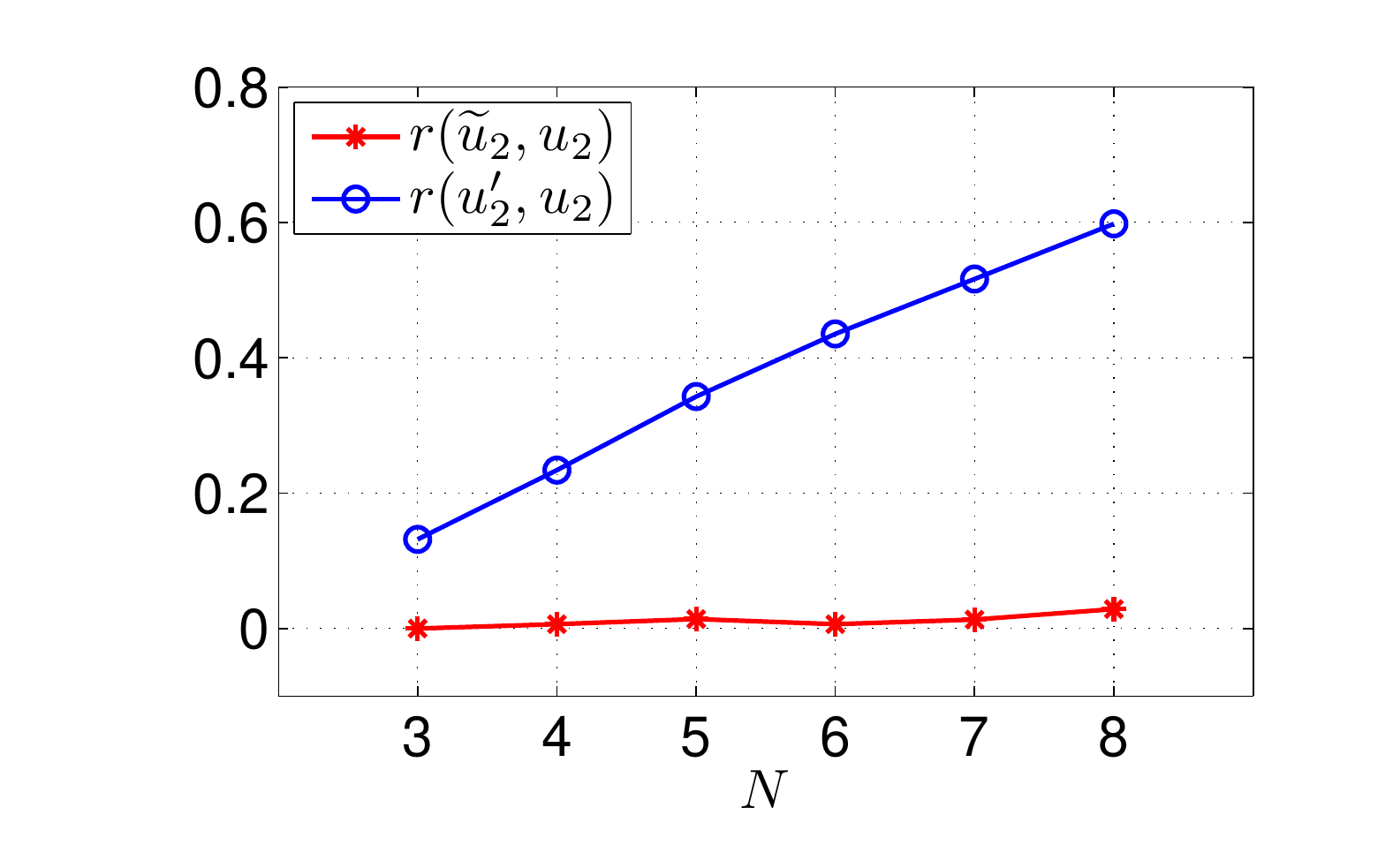}}
    \subfloat[]{\includegraphics[width=.5\textwidth]{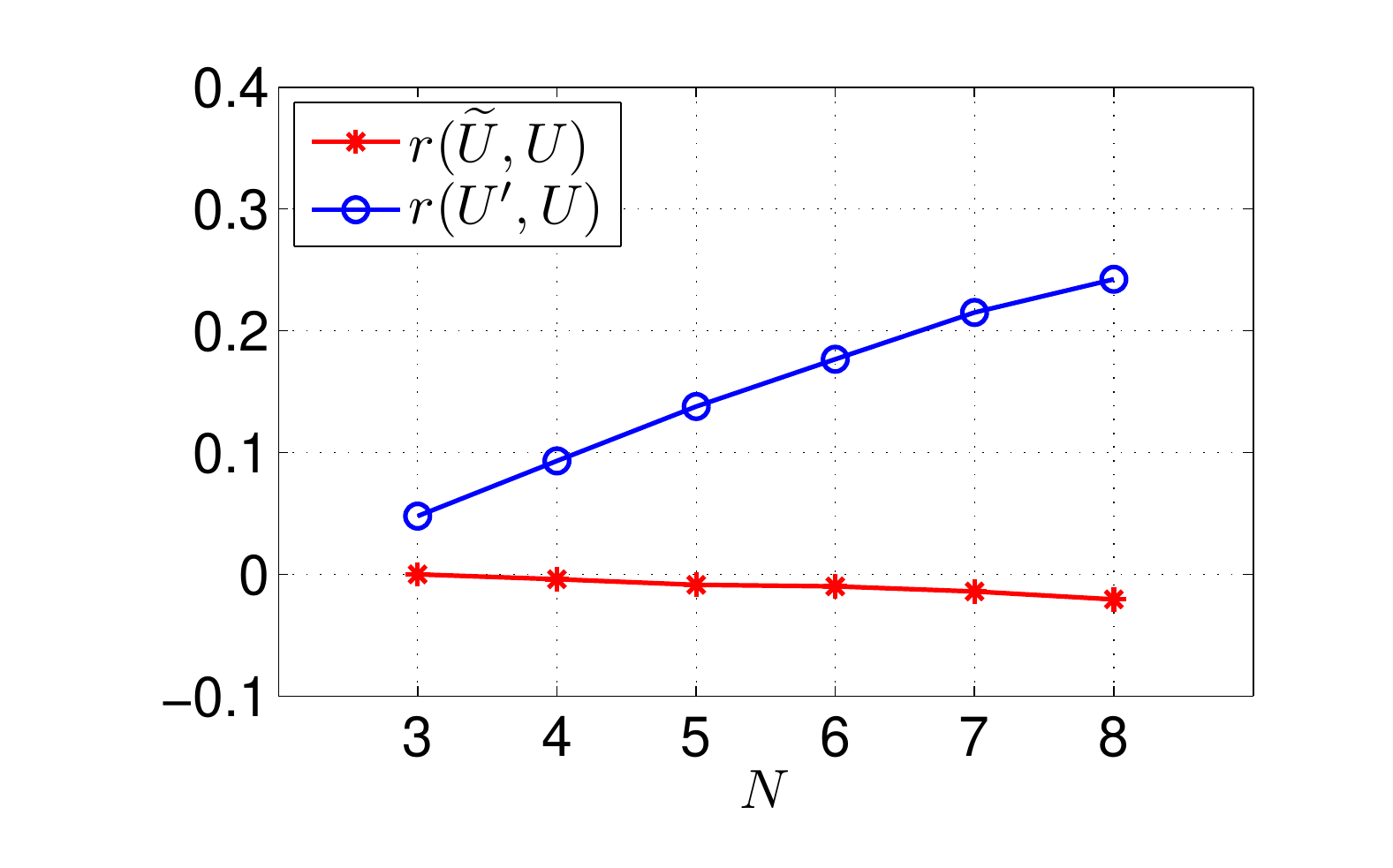}}
    \caption{Comparison of variation between different bases. }\label{fig_ru2}
\end{figure}

We also compare the sum of variations of the two bases.
Denote
\[
    S(U):=\sum_{k=1}^N S(u_k), \quad S(\widetilde U):=\sum_{k=1}^N S(\widetilde u_k)
\]
and relative error
\[
    r(\widetilde U,U):=\frac{S(\widetilde U) - S(U)}{S(U)}.
\]
The average of $r(\widetilde U, U)$ for $100$ random graphs is plotted in Figure \ref{fig_ru2}(b) by the red line.
The relative error $r(U',U)$ between $S(U')$ and $S(U)$, where $U'$ is the Laplacian basis, is also plotted for the sake of completeness (blue line in Figure \ref{fig_ru2}(b)).
It can be seen that $r(\widetilde U,U)$ is below zero, i.e., the sum of variation of $\widetilde U$ is even smaller than that of $U$.
That means, if one considers the problem of minimizing the sum of variation of the whole basis, i.e.
\[
    \begin{array}{cl}
    \min\limits_{U\in\R^{N\times N}}& S(U)\\
    \text{s. t.}& U^\top U = I,
    \end{array}
\]
then the greedy basis $\tilde U$ might give a better approximate solution than the $\ell_1$ basis.

\subsection{$n$-term approximation}
A nice property of the classical Fourier transform is that the Fourier coefficient usually has a fast decay for most real-world signals.
That means one can drop the high frequency coefficients without losing much information, which serves as the foundation of various signal compression methods.
In our second experiment, we will examine this property for the greedy basis $\widetilde U$, and compare it to the Laplacian basis $U'$.

Given a signal $x$, let the Fourier transform under the Laplacian basis be denoted by $\hat x' = U'^\top x$, and the Fourier transform under the greedy basis be denoted by $\widetilde{\hat x} = \widetilde U^\top x$.
Suppose we use the largest $n$ terms of coefficients to reconstruct $x$.
Namely we sort the coefficients in descending order, say $|\widetilde{\hat x}(k_1)|\ge \cdots \ge |\widetilde{\hat x}(k_N)|$, for the greedy basis.
Then we define the $n$-term approximation
\[
    \widetilde y_n := \sum_{i=1}^n \widetilde{\hat x}(k_i) \widetilde u_{k_i}
\]
and the approximation error
\[\widetilde\varepsilon_n :=  \frac{\|x- \widetilde y_n\|} {\|x\|} = \frac{\big(\sum_{i = n+1}^N|\widetilde{\hat x}(k_i)|^2\big)^{1/2}}{\|\hat x\|}.\]
For the Laplacian basis, we define the $n$-term approximation $y_n'$ and error $\varepsilon_n'$ in a similar way.

The experiment is performed on two signals.
The first is a simulated signal, defined through its Fourier coefficients under the Laplacian basis:
\[
    \hat x'(k) := \frac{1}{1+\mu \lambda_k} \times {\rm rand}(k),
\]
where $\mu$ is a constant, $\lambda_k$ is the Laplacian eigenvalue, and ${\rm rand}(k)$ is a random number uniformly distributed on $[-1,1]$.
Figure \ref{fig_approx} (top) plots the simulated signal (left), its Fourier coefficients under the two bases (middle) and the corresponding approximation errors (right).

The second example is a real-world signal: the average temperature of Switzerland during 1981-2010 \cite{swiss}.
See Figure \ref{fig_approx} (bottom) for the results.
It can be seen that for either simulated or real-world signal, both types of Fourier transform lead to a fast decay of approximation error, and the rates of decay are almost the same.
\begin{figure}[h]
\centering
  \includegraphics[width=\textwidth]{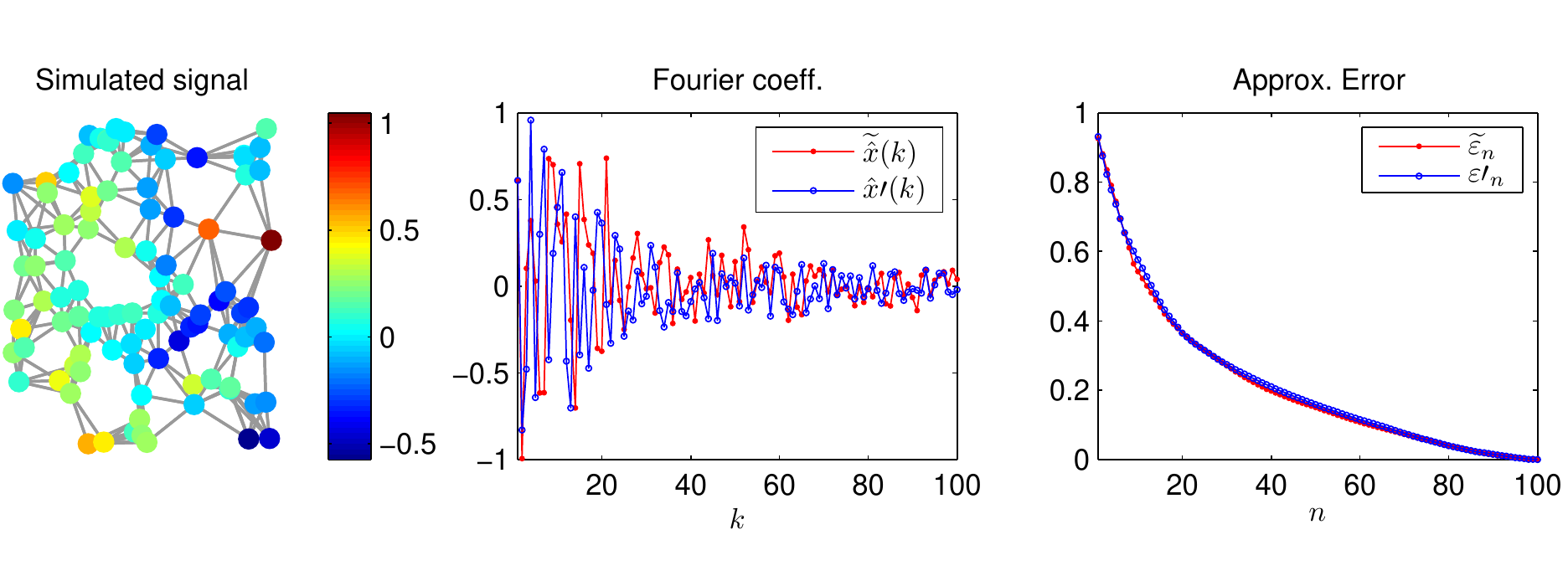}
  \includegraphics[width=\textwidth]{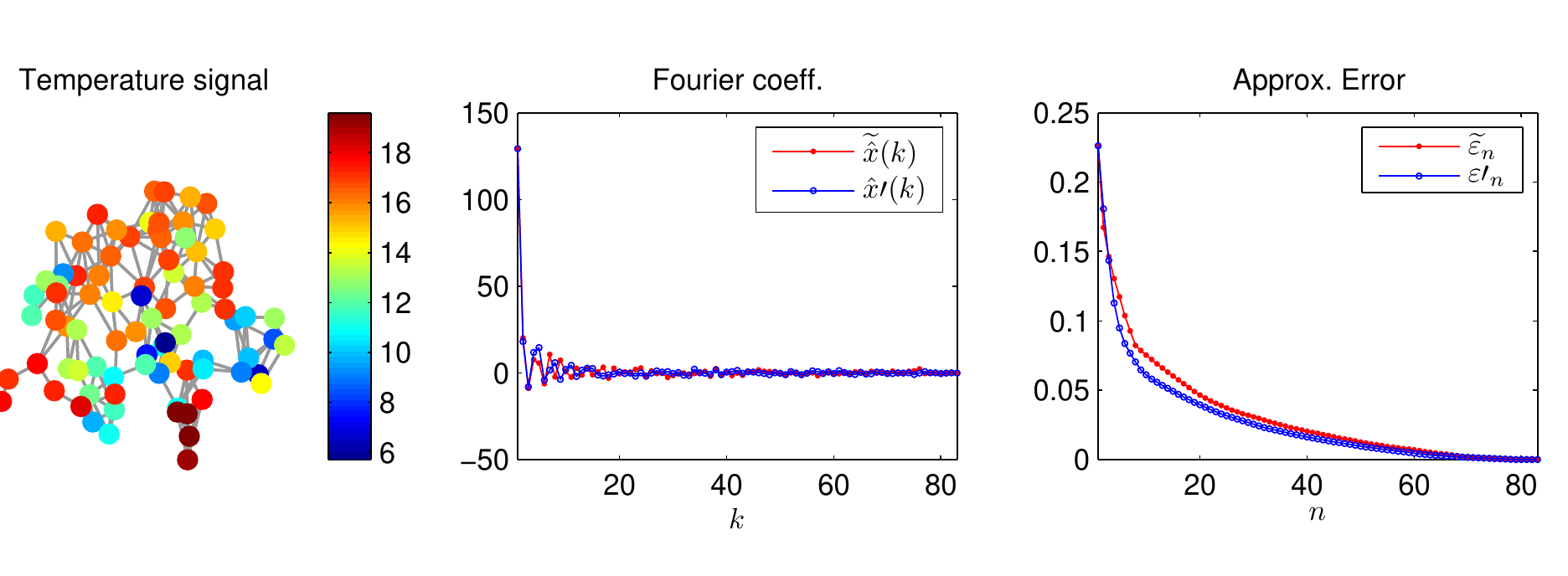}
  \caption{$n$-term approximation under two bases. Red: Greedy basis. Blue: Laplacian basis. }\label{fig_approx}
\end{figure}

\section{Conclusion}
In this paper we propose a definition of $\ell_1$ Fourier basis of a graph as the solutions of a sequence of $\ell_1$ norm variation minimization problems.
We obtains a necessary condition satisfied by the local minimum, which implies the number of values of $u_k$ is at most $k$.
Furthermore, we show that there are finitely many isolated local minima, contained in a finite set $\X^*_U$, and it is possible to enumerate $\X^*_U$ to find the global minimum when $N$ is small.
For large $N$, we give a fast greedy algorithm to approximately construct the $\ell_1$ basis, based on a greedy partition sequence created by grouping the vertices according to their mutual weights.
Numerical experiments show that the greedy basis provides a good approximation to the $\ell_1$ basis.
Also, the Fourier transforms of the two bases (greedy basis and Laplacian basis) have the same rate of decay for simulated or real signals.
As for future directions, we suggest considering the general $\ell_p$ norm variation minimization problem and the corresponding $\ell_p$ norm Fourier basis.

\section*{Acknowledgements}
This work is supported by National Natural Science Foundation of China (Nos. 11601532, 11771458, 11501377, 11431015).

\end{document}